%% file: optimality_clue.tex
\documentclass[a4paper,9pt]{scrartcl}

\usepackage{graphicx}
\usepackage[T1]{fontenc} 

\usepackage{longtable}
\usepackage{booktabs}
\usepackage{soul, color} 
\usepackage{epsfig} 
\usepackage{xcolor}
\usepackage{multirow}
\usepackage{algorithm,algorithmic}
\usepackage[algo2e,vlined]{algorithm2e} 
\SetKwInOut{Input}{Data}
\SetVlineSkip{1pt}

\usepackage{placeins}
\usepackage{amssymb}         
\usepackage{amsmath}
\graphicspath{{../images/}}
\usepackage{comment}
\usepackage[a4paper]{geometry} 
\usepackage{ragged2e} 
\justifying

\newcommand{\intervalleff}[2]{\mathopen{[}#1\,;#2\mathclose{]}}

\newcommand\keywords{}
\definecolor{lightblue}{rgb}{0.8,0.9,1} 
\definecolor{lightred}{rgb}{1,0.5,0.4} 

\usepackage[normalem]{ulem}

\newtheorem{Definition}{Definition}
\newtheorem{Thm}{Theorem}
\newtheorem{Lem}{Lemma}
\newtheorem{Coro}{Corollary}
\newenvironment{proof}{\textbf{Proof}}{\begin{FlushRight}
                                        $\square$
                                       \end{FlushRight}}

\begin{document}
\title{Optimality Clue for Graph Coloring Problem}
%
%
\author{Alexandre Gondran\\ENAC, French Civil Aviation University, Toulouse, France\\\texttt{alexandre.gondran@enac.fr}
\and
Laurent Moalic\\UHA,  University of Upper Alsace, Mulhouse, France\\\texttt{laurent.moalic@uha.fr}
}
%

%
%
\date{}
\maketitle              
\input{optimality_abstract}


\input{optimality_sections}

\input{optimality_appendix}
\bibliographystyle{splncs04}
\bibliographystyle{elsarticle-num}
\bibliography{gcp_bib,counting_bib}

\end{document}

%% file: optimality_abstract.tex
\begin{abstract}\small\baselineskip=9pt
In this paper, we present a new approach which qualifies or not a solution found by a heuristic as a potential optimal solution.
Our approach is based on the following observation: for a minimization problem, the number of admissible solutions decreases with the value of the objective function.
For the Graph Coloring Problem (GCP), we confirm this observation and present a new way to prove optimality. 
This proof is based on the counting of the number of different $k$-colorings and the number of independent sets of a given graph $G$.

Exact solutions counting problems are difficult problems (\string#P-complete).
However, we show that, using only randomized heuristics, it is possible to define an estimation of the upper bound of the number of $k$-colorings. 
This estimate has been calibrated on a large benchmark of graph instances for which the exact number of optimal $k$-colorings is known.


Our approach, called \emph{optimality clue}, build a sample of $k$-colorings of a given graph by running many times one randomized heuristic on the same graph instance.
We use the evolutionary algorithm \emph{HEAD}~\cite{Moalic2018}, which is one of the most efficient heuristic for GCP. 

\emph{Optimality clue} matches with the standard definition of optimality on a wide number of instances of DIMACS and RBCII benchmarks where the optimality is known. 
Then, we show the clue of optimality for another set of graph instances.

\keywords{Optimality  \and Metaheuristics \and Near-optimal.}

\end{abstract}

%% file: optimality_sections.tex
\section{Introduction}



For a given integer $k\geq1$, a $k$-coloring  of a given graph $G=(V,E)$ is an assignment of one of $k$ distinct colors to each vertex $v\in V$ in the graph, so that no two adjacent vertices (linked by an edge $e\in E$) are given the same color. 
The Graph Coloring Problem (GCP) is to find, for a given graph $G$, its \textit{chromatic number} $\chi(G)$ corresponding to the smallest $k$ such that there exists a $k$-coloring of $G$. 
GCP is NP-hard~\cite{Karp72} for $k\geq 3$. The $k$-coloring problem ($k$-CP) is the associated decision problem.
For an optimization problem which is NP-hard, there is no efficient polynomial-time exact algorithm to solve it, unless P$=$NP.
Therefore for large size instances of a minimization NP-hard problem, the exact algorithms must be stopped before their end.
In this case, exact algorithms such as branch and bound methods find a lower bound of the optimal value of the objective function.
Heuristic approaches are then the only ways to find, in reasonably fast running-time, a ``good'' solution in terms of objective function value, i.e. an upper bound of the optimal value. 
However, even if an admissible solution is found, its distance to the optimal solution remains unknown, except for approximation algorithms\footnote{Notice that it is still NP-hard to approximate $\chi(G)$ within $n^{1-\epsilon}$ for any $\epsilon > 0$~\cite{Zuckerman2007}.}. 
The optimality gap is the different between the upper bound (found by a heuristic) and the lower bound (found by a partial exact method). 
Optimality is proven only when this gap is equal to zero. 
Unfortunately for large size instances of an NP-hard problem, this gap is often important. 
It is particularly true for challenging instances~\cite{Held2012,Moalic2018} of the GCP of the DIMACS benchmark~\cite{dimacs96}.
This paper addresses the following question: 
What to do in this situation?
Is it possible to prove optimality of a graph coloring problem instance 
using only heuristic algorithms?

The response is \emph{Yes}, for specific class of graphs: for example, it exists efficient polynomial-time exact algorithms to find $\chi(G)$ for
 \emph{interval graphs}, \emph{chordal graphs}, \emph{cographs}~\cite{Orlin1981,Shih1989}.
For some graphs like \emph{1-perfect graphs}\footnote{A perfect graph is a graph in which the chromatic number of every induced subgraph equals the size of the largest clique of that subgraph. 1-perfect graphs are more general than perfect graphs. There exists polynomial-time exact algorithms to find $\chi(G)$ for perfect graphs~\cite{Grotschel1984}, but slow in practice.  Line graphs, chordal graphs, interval graphs or cographs are subclasses of perfect graphs.}, for which the chromatic number $\chi(G)$ is equal to the size of the maximum cliques $\gamma(G)$, 
it is possible to solve the dual problem, the Maximum Clique Problem (MCP), with another heuristic and conclude to optimality if 
the size of the maximum clique found is equal to the smallest number of colors used for coloring $G$ found also by a heuristic. 
In this specific case, the optimality gap (or duality gap between GCP and MCP) is zero.

However, the response to the question is \emph{No}, in general case; a heuristic finds approximate solutions (upper bound); although the coloring found may be optimal, it is not possible to prove this possible optimality. 
Therefore, the question become: what can be done better using only a heuristic than finding an approximate solution?
Is it possible to define a kind of optimality index for a graph coloring problem instance?

One shows in this article that a heuristic does not only find an upper bound of $\chi(G)$ but that it is also able to count the number of different $k$-colorings (i.e. the number of admissible solutions having the same objective function value).
Our approach is based on the fact that the number of different $k$-colorings decreases dramatically when the number of colors, $k$, decreases too.
Indeed figure~\ref{fig:nbsol} gives a typical example of a random graph with 30 vertices, a density of $0.9$ and $\chi(G)=16$. 
The number of colorings with exactly $k$ colors (blue bars) and the total colorings with $k$ colors or less (red bars), noted $\mathcal{N}(G,k)$, are exactly computed for all values from $k=16$ to $k=30$. $\mathcal{N}(G,k)$ decreases exponentially when $k$ decreases to $\chi(G)$.
One proves a theorem showing that when the number of $k$-colorings is lower than a given value (the number of independent sets of $G$~\footnote{An independent set is a subset of vertices of $G$, such that every two distinct vertices in the independent set are no adjacent.}), then we achieve the optimum: $\chi(G)=k$.

\begin{figure}
 \begin{minipage}[b]{.48\linewidth}
  \centering\epsfig{figure=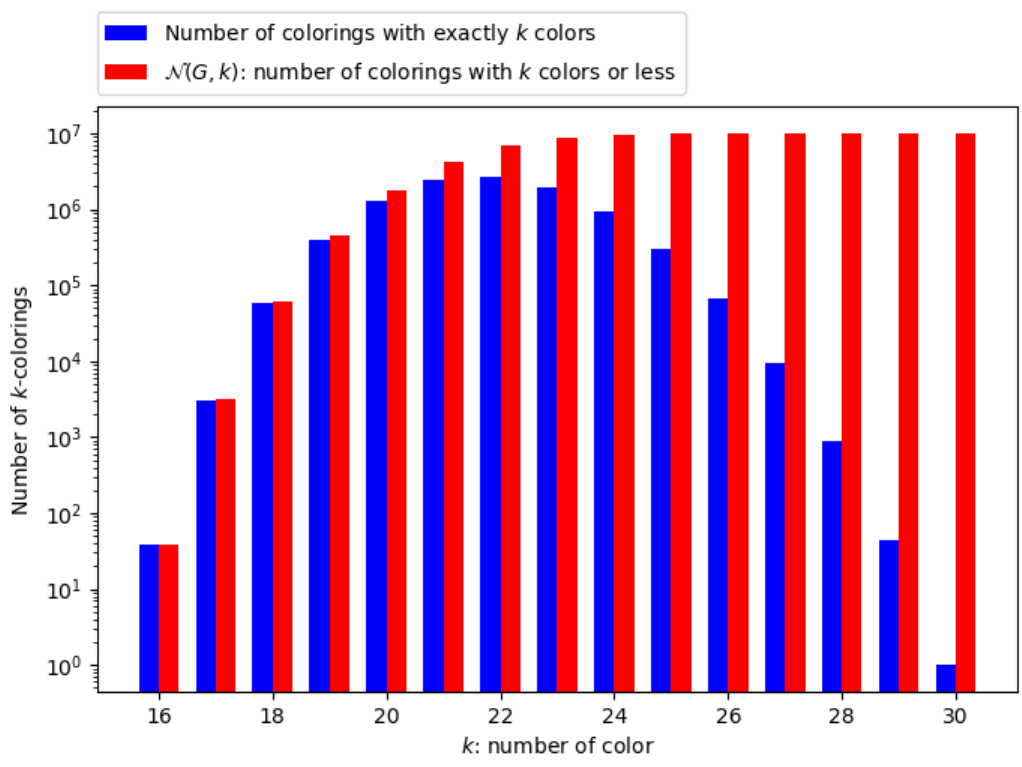,width=\linewidth}
  \caption{\label{fig:nbsol}Number of colorings with exactly $k$ colors (blue bars) and number of total colorings with $k$ colors or less (red bars), noted $\mathcal{N}(G,k)$ in function of $k=16...30$, for a random graph with 30 vertices, density $0.9$ and $\chi(G)=16$.}
 \end{minipage} \hfill
 \begin{minipage}[b]{.48\linewidth}
  \centering\epsfig{figure=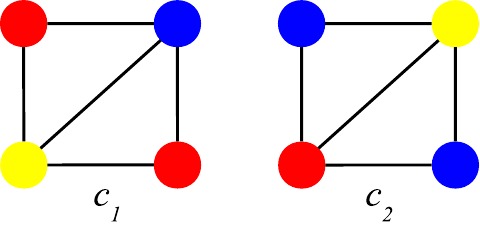,width=0.6\linewidth}
\caption{\label{fig:permutation}Two $3$-colorings $c_1$ and $c_2$ of the same graph with four vertices. These two colorings have to be considered identical because $d(c_1,c_2)=0$ with $d$ the set-theoretic partition distance; we pass from one to the other, just by a permutation of color classes.}
 \end{minipage}
\end{figure}



In this article, we try to apply the proposed theorem in order to prove optimality. 
\subsubsection*{Brief solutions counting review}
Our work tackles the problem of counting solutions of NP-complete problems which has been widely studied for boolean SATisfiability problem, called \string#SAT, or Constraint Satisfaction Problem (CSP), called \string#CSP; $k$-coloring problem is a special case of CSP. These problems are known as \string#P-complete~\cite{Valiant1979}.
A recent survey on \string#CSPs is done in~\cite{Jerrum2017}. 
Even if a problem is not NP-hard, the problem of solutions counting is often hard. Specific studies on counting solutions of $k$-CP are done in~\cite{Jerrum1995,Favier2011,Miracle2016}. Because the exact counting is in many cases a complex problem, statistical or approximate counting are often considered. Then, uniform sampling of the set of solutions problem is related to the problem of counting solutions. 
Many works are done on uniform or near uniform sampling like~\cite{Gomes2007,Gomes2009,Wei2005}. The objective is to count by sampling.
Frieze and Vigoda~\cite{Frieze2007} give a survey on the use of Markov Chain Monte Carlo algorithms for approximately counting the number of $k$-colorings. The features of ergodicity or quasi-ergodicity of the heuristics that guarantee an uniform sampling are deeply discussed in~\cite{Ermon2012}. However, theoretical results are obtained with a high value of $k\geq\Delta$ where $\Delta$ is the maximum degree of the graph $G$ which is very far from $\chi(G)$ for challenging graphs. 
On the other hand, when tests are performed with $k=\chi(G)$ like in~\cite{Favier2011}, the considered graph instances are often with more than $10^{20}$ $k$-colorings. 
If the number of $k$-colorings is too high (higher than the number of independent sets), then it is not possible to apply our theorem. 
Therefore, in practice, our approach can be applied to graphs that do not have too many optimal colorings; we considered graphs with at most 1 million different \emph{optimal} colorings.

To our knowledge, it is the first time that solutions counting are used to prove optimality.
We define a procedure, called \emph{optimality clue}, in order to apply the proposed theorem.
First, we build a sample of $k$-colorings of a given graph $G$ by running many times (about 1,000 times) the same randomized heuristic algorithm. 
In this study, we use \emph{HEAD}\footnote{Open-source code available at: github.com/graphcoloring/HEAD}, 
our open-source memetic algorithm (i.e. hybridization of tabu search and evolutionary algorithm),
which is very efficient heuristic solving GCP~\cite{Moalic2018}. 

In this sample some colorings may appear several times and others only ones. 
The number of different $k$-colorings inside the sample is used to build an estimation of the total number of colorings with $k$ colors. 
This estimator has been calibrated on a large benchmark of graph instances for which the number of optimal $k$-colorings is exactly known.
Because we have no guarantee that the sampling is uniform, in the general case, therefore we have no guarantee that our estimator is always exact.

Moreover, building a sample of $k$-colorings is time-consuming, then the size of the sample should be ``reasonable''.
Therefore, graphs for which our optimality clue can be calculated are graphs having not too many optimal $k$-colorings (i.e. about less than one million). 
Of course it is not possible to known a priori if a given graph has more or less than 1 million optimal colorings.
Then, our approach provides a clue that a coloring found by the heuristic is perhaps optimal (positive conclusion) but never denies it (no negative conclusion): 
in many cases we can not have any conclusion.

This article is organized as follows. In Section~\ref{sec:proof}, we present the new optimality proof for GCP based on solutions counting. 
Our general approach, called \emph{optimality clue}, is define in Section~\ref{sec:clue}. 
In Section~\ref{sec:ub}, we detail how we calculate the estimate of the number of $k$-colorings using benchmark graph instances.  
Numerical tests and experiments are presented in Section~\ref{sec:test}.
Finally, we conclude in Section~\ref{sec:conclusion}. 

\section{Proof of Optimality by solutions counting}\label{sec:proof}

Notice that there are different ways to count the $k$-colorings of a given graph $G$.
When counting the number of different $k$-colorings, we have to take into account the permutations of the color classes. We consider one $k$-coloring not as an assignment of one color among $k$ to each vertex but as a partition of the vertices of the graph into $k$ independent sets. 
An Independent Set (IS) or stable set is a set of vertices of $G$, no two of which are adjacent.
Two $k$-colorings $c_1$ and $c_2$ are considered identical if they correspond to the same partition of $G$. The distance between two $k$-colorings that is taken into account is  the \textit{set-theoretic partition distance} used in~\cite{Galinier99,Gusfield2002,Moalic2018}, which is independent of the permutation of the color classes.
In previous works about solutions counting of $k$-CP~\cite{Favier2011}, authors counted
the total number of $k$-colorings including all the permutations like in the example of Figure~\ref{fig:permutation}; such a calculation of the number of different $k$-colorings is $k!$ times higher than the way we count. 
This makes their methods inapplicable to our study.
We write $\Omega(G,k)$ the set of all $k$-colorings of the graph $G$. A $k$-coloring can use exactly $k$ colors or less, then $\Omega(G,k-1)\subset \Omega(G,k)$.
The cardinal of $\Omega(G,k)$ is noted $\mathcal{N}(G,k)=|\Omega(G,k)|$.



Our approach is based on the following fact~:

\begin{Lem}\label{lemme:iG}
Let a graph $G$ and an integer $k\geq1$. 
If there exists at least one $k$-coloring of $G$, then there exists at least $i(G)-k+1$ different $(k+1)$-colorings of $G$: $$\mathcal{N}(G,k+1)\geq\mathit{i}(G)-k+1,$$ where $i(G)$ is the number of independent sets of $G$. 
\end{Lem}

\begin{proof}
Notice that a $k$-coloring of a graph $G=(V,E)$ is a partition of $|V|$ vertices into $k$ IS.
Indeed vertices colored with the same color inside a $k$-coloring are necessarily an IS.
In other words, it is always possible to color all vertices of any IS with the same color.
We note $IS(G)=\{U\subset V\ |\ \forall x,y\in U^2,\ \{x,y\}\notin E\}$ the set of all the IS of $G$, then $i(G)=|IS(G)|$.

Starting with one coloring of $G$ with exactly $k$ colors, 
for each independent set of $G$ except for the $k$ IS of the $k$-coloring, 
it is possible to recolor all vertices of this independent set with a new color (the $(k+1)$th color). 
We obtain by this way one different $(k+1)$-coloring for each different independent set, then we count at least a total of $\mathit{i}(G)-k$ different colorings with exactly $(k+1)$ colors. Then, $\mathcal{N}(G,k+1)\geq\mathit{i}(G)-k+1$ because we have to count also the starting $k$-coloring. 

\end{proof}

Then, we obtain the following theorem:

\begin{Thm}
\label{thm:opt}
Let a graph $G$ and an integer $k\geq1$. 
Let $\mathcal{N}(G,k)$ the number of $k$-colorings of $G$ and $\mathit{i}(G)$ the number of independent sets of $G$. 

If $\mathit{i}(G)-k>\mathcal{N}(G,k)>0$, then $\chi(G)=k$.
\end{Thm}

\begin{proof}
 $\chi(G)\leq k$ because $\mathcal{N}(G,k)>0$. If $\chi(G)<k$, it means that there exists at least one $(k-1)$-coloring (i.e. $\mathcal{N}(G,k-1)>0$). If we add a new color, it is possible to consider this $(k-1)$-coloring and to recolor any independent set of $G$ with the new color, we obtain by this way $\mathit{i}(G)-k$ different $k$-colorings (by Lemma~\ref{lemme:iG}). Therefore $\mathit{i}(G)-k\leq\mathcal{N}(G,k)$ which refute initial assumption. 
\end{proof}

For example, the studied graph in Figure~\ref{fig:nbsol} (30 vertices and density 0.9) has 38 different colorings with 16 colors: $\mathcal{N}(G,k=16)=38$; moreover this graph has 78 IS: $i(G)=78$, then the theorem is applicable with $k=16$ because: $\mathit{i}(G)-k=78-16=62>38=\mathcal{N}(G,k)>0$. Then, thanks to the theorem we can conclude that $\chi(G)=16$. Moreover, for $k=17$, $\mathcal{N}(G,k=17)=3121 > i(G)-k=61$, so the theorem is not applicable.

\begin{Coro}\label{coro:opt}
Let a graph $G$ and an integer $k\geq1$. 
Let $\overline{\mathcal{N}}(G,k)$ an upper bound of the number of $\mathcal{N}(G,k)$ and $\underline{\mathit{i}}(G)$ a lower bound of $\mathit{i}(G)$. 

If $\mathcal{N}(G,k)>0$ and $\underline{\mathit{i}}(G)-k>\overline{\mathcal{N}}(G,k)$, then $\chi(G)=k$.
\end{Coro}

\section{Optimality Clue}\label{sec:clue}
We propose in this paper to apply the corollary~\ref{coro:opt}, 
so to find an appropriate upper bound of the number of $k$-colorings of $G$, 
$\overline{\mathcal{N}}(G,k)$, and a lower bound of the number of independent sets of $G$, 
$\underline{\mathit{i}}(G)$. 
\subsection{IS counting}
There exists many algorithms~\cite{Bron1973,Carraghan1990,Ostergard2002,Samotij2015} for counting all the maximal independent sets of a graph $G$ (or similarly counting all the maximal cliques\footnote{A maximal clique is a clique that cannot be extended by including one more adjacent vertex. 
A maximum clique is a clique that has the largest size in a given graph; a maximum clique is therefore always maximal, but the converse does not hold. Analogue definition for IS.} in $\overline{G}$, the complementary graph of $G$). By definition, the number of maximal IS, noted $i_{max}(G)$, is a lower bound of $i(G)$. Those algorithms are based on enumeration. 
Because we focus this study on graphs having less than 1 million optimal solutions, we can stop the enumerating after finding 1 million IS.
Generally, $\mathit{i}(G)$ is very high except for graphs with very high density. 
Real-life graphs have often a low density, then $\mathit{i}(G)$ is very high. Moreover, a simple lower bound is given by~\cite{Pedersen2006}~: $\mathit{i}(G)\geq2^{\alpha(G)}+n-\alpha(G)$, where $\alpha(G)$ is the size of the largest independent set of $G$ and $n$ the number of vertices. 
Bollob\'{a}s' book~\cite{Bollobas2001}~(p.283) gives also a statistical number of maximal cliques of size $p$ for a random graph. Then, we conclude that:
$$
i_{max}(G) \approx i_B(G) = \sum_{p=1}^{n}\binom{n}{p}(1-d)^{\binom{p}{2}}
$$
with $n$ the number of vertices and $d$ the density of a random graph $G$. 

In this study, we use Cliquer~\footnote{Code available at: users.aalto.fi/~pat/cliquer.html. To count all IS of a graph, you just execute: ./cl <complement graph> -a -m 1 -M <k>}, an exact branch-and-bound algorithm developed by Patric \"{O}sterg\aa{}rd~\cite{Ostergard2002} that enumerates all cliques (an IS is a clique in the complementary graph). 

It is more complex to evaluate $\overline{\mathcal{N}}(G,k)$ and 
section~\ref{sec:ub} presents a way to build an \emph{experimental} 
upper bound of $\mathcal{N}(G,k)$.
We characterize this upper bound as \emph{experimental} because 
it is based on experimental tests on benchmark graph instances, 
then there is no total guaranty that it is an upper bound.

\subsection{Procedure}\label{sec:procedure}

We define here the procedure of what we call \emph{Optimality Clue} for graph coloring:
let $G$ a graph and $k>0$ a positive integer, that we suspect to be the chromatic number of $G$.
The proposed approach is based on the five following steps:
\begin{enumerate}
\item Build a sample of $t=1,000$ $k$-colorings of $G$: we run the memetic algorithm \emph{HEAD} 
 on $G$ as many times as needed to obtain $t$ legal $k$-colorings. Those solutions are \emph{the solutions sample}. The size of the sample is equal to $t$. We take in general case $t=1,000$ when it is possible.
\item Count the number of different $k$-colorings inside the sample. This number is equal to $p$. Of course $0\leq p\leq t$.
\item Estimate an upper bound of $\mathcal{N}(G,k)$ as $UB(p,t)$ (cf. Section~\ref{sec:ub}); this upper bound is function of $t$ and $p$.
\item Compute $i(G)$, the number of IS, or at least a lower bound if $i(G)>10^6$, with an exact algorithm (Cliquer). 
\item If $i(G) > UB(p,t)$, then we conclude that solutions of the sample have a \textbf{clue to be optimal}:
\begin{center}
Chances are that $k$ is equal to $\chi(G)$
\end{center}
\end{enumerate}

\subsubsection{Uniform sample}
If the sample is uniform\footnote{All $k$-colorings in the sample are uniformly drawn at random in $\Omega(G,k)$.}, then there exists statistical methods to count solutions and to build an upper bound with statistical guarantee, for example the capture-recapture methods: Peterson method~\cite{Krebs2009}, Jolly-Seber method~\cite{Baillargeon2007} wich is commonly used in ecology to estimate an animal population's size.
However, it is not our case: we have no guarantee that our solutions sample is uniform or near uniform.
HEAD is a memetic algorithm that explores the space of non-legal $k$-colorings: 
a non-legal $k$-coloring is a coloring with at most $k$ colors and where two adjacent vertices (linked by an edge) may have the same color (called conflicting edge). 
The objective of HEAD is to minimize the number of conflicting edges to zero, that is to get a legal $k$-coloring.
HEAD is an evolutionary algorithm with a population size equals to two. The two non-legal $k$-colorings perform at each generation a tabu search and after a crossover.
The sample distribution depends on the fitness landscape properties~\cite{Merz2000,Marmion2013}\footnote{The fitness landscape itself depends on the neighborhood used for tabu search and the crossover used.} and there is no reason for this distribution to be uniform.
A smooth landscape (respectively a rugged landscape) around a legal $k$-coloring will increase (resp. decrease) the probability of finding this $k$-coloring.
Figure~\ref{fig:ergodicity} represents the frequency of the 319 optimal $46$-colorings of <r140\_90.4> graph of RCBII benchmark (140 vertices and density 0.9) in a sample of size 100,000 found by HEAD heuristic. 
In this typical graph instance, 
the ratio between the least frequent and the most frequently found coloring is around a factor of $10^3$ which corresponds to the same scale as similar studies~\cite{Wei2005}.
\begin{figure}
 \begin{minipage}[b]{.48\linewidth}
  \centering\epsfig{figure=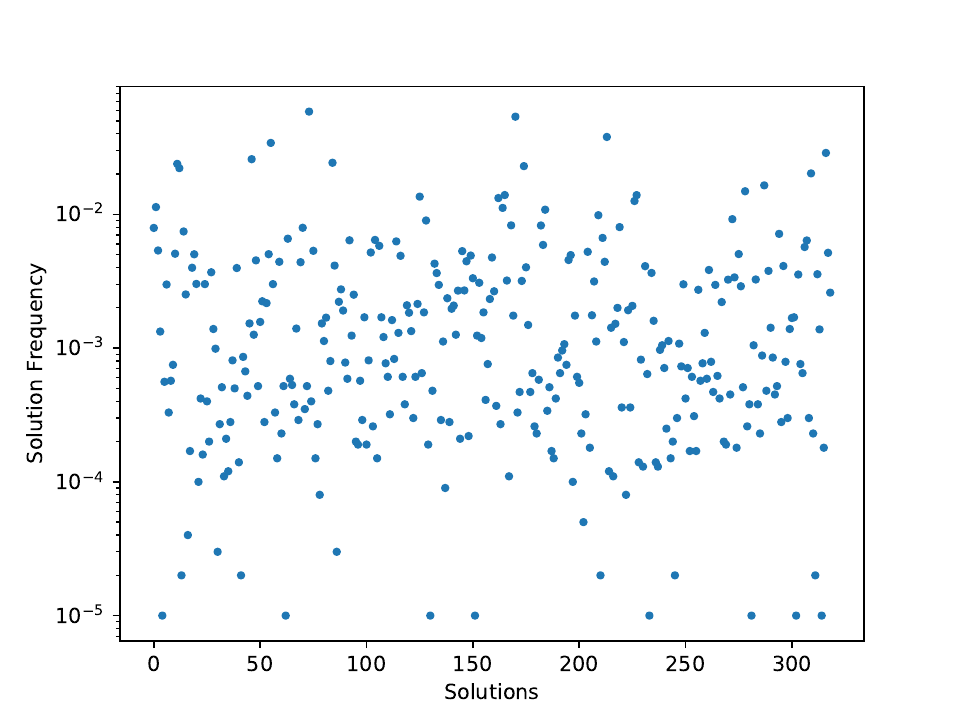,width=\linewidth}
\caption{\label{fig:ergodicity}Sampling of $46$-colorings for the <r140\_90.4> graph from RCBII benchmark (140 vertices and density 0.9).}
 \end{minipage} \hfill
 \begin{minipage}[b]{.48\linewidth}
  \centering\epsfig{figure=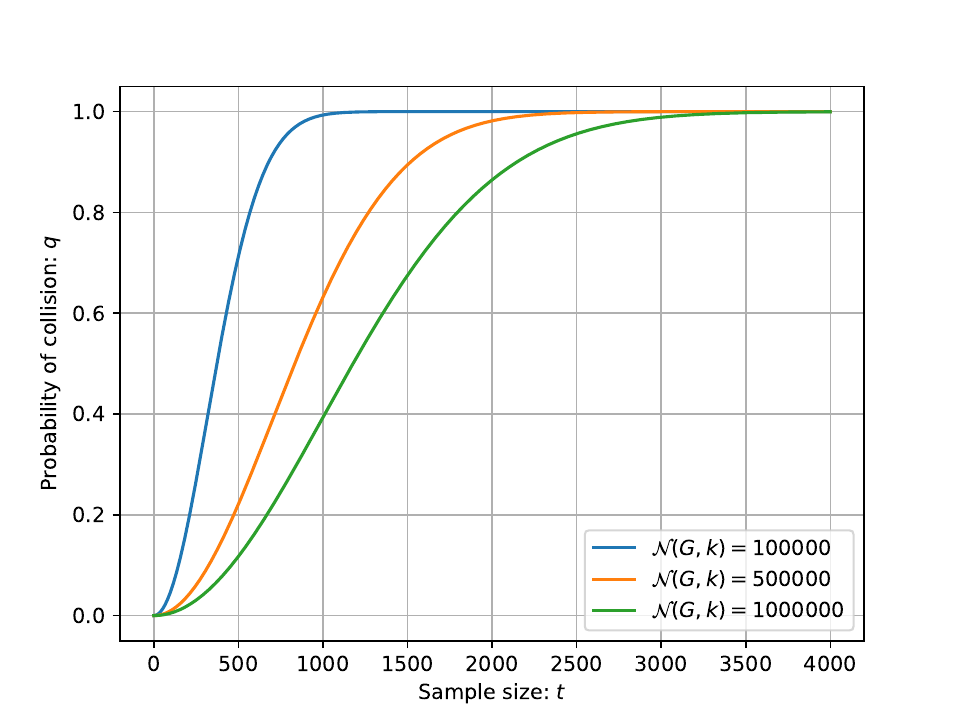,width=\linewidth}
\caption{\label{fig:CollisionProbability}Collision probability $q$, given the sample size $t$, and the total number of $k$-colorings $\mathcal{N}(G,k)$.}
 \end{minipage}
\end{figure}


Another approach is to take into account the ergodicity of an algorithm, which is its capability to explore all the search space.
More precisely, an algorithm is ergodic if it is possible (probability not null) to reach any $k$-coloring from any other $k$-coloring in a finite number of iterations. 
Random walks or Metropolis algorithms (with a positive temperature sufficiently high) are ergodic algorithms since there is always a finite probability of escaping from local minimum.
However, those algorithms are very inefficient in practice to find an optimal $k$-coloring in the general case.





\subsubsection{Sample size}
The choice of $t$, the size of the sample, is very important for two reasons.
First, in practice, to build a sample of $k$-colorings can be very time-consuming, 
then the size of the sample should have a reasonable size. We take $t=1,000$ for most of the graph instances.
However, the more challenging the graph instance, the longer HEAD takes to find one $k$-coloring.
Therefore, it is not possible to build a sample of size 1,000 for all graphs, such as for the <DSJC500.5> graph of DIMACS (cf. Table~\ref{tab:datasetDIMACS}).

The second reason is more theoretical.
We have limited the maximum number of different optimal solutions to 1 million, for a graph to be considered by our approach.
In fact, we choose 1 million because it equals to $t^2$ with $t=1,000$.
Indeed, if the sample is uniformly drawn at random in $\Omega(G,k)$,
the probability $q$ that at least two colorings of the sample are identical is equal 
to\footnote{This problem is linked to the birthday problem that shows
that in a room of just 23 people there's a 50-50 chance that two people have the same birthday.
In our case, the number of days in a year is $\mathcal{N}$ and 
the number of people is the size $t$ of the sample.}: 
$q=1-\frac{\mathcal{N}!}{\mathcal{N}^t(\mathcal{N}-t)!}\simeq1-e^{-\frac{t(t-1)}{2\mathcal{N}}}$ 
then $\mathcal{N}\simeq-\frac{t(t-1)}{2ln(1-q)}$. 
We call also $q$ the collision probability.
So, if $q=0.5$ then $\mathcal{N}\simeq720626$, if $q=0.393$ then $\mathcal{N}\sim t^2=10^6$.
Figure~\ref{fig:CollisionProbability} represents the collision frequency, $q$, in function of the sample size, $t$, for different values of the $\Omega(G,k)$ size.
When $\mathcal{N}(G,k)=10^5$ and $t=1,000$, it is almost impossible to miss a collision in the sample, but for $\mathcal{N}=10^6$, there is around 60\% to miss a collision.
However, it is not tragic to miss a collision for our approach. 
Indeed, the consequence is that the clue of optimality may be not applicable but the risk of \emph{false positive} is avoided. 
A false positive occurs if our procedure~\ref{sec:procedure} improperly indicates the optimality clue, when in reality the $k$-colorings are not optimal.
Moreover, the collision frequency is higher for a non-uniform sample than for a uniform one.

\section{Estimate of the number of $k$-colorings: $UB(G,k,p,t)$}\label{sec:ub}

\subsection{Data sets}






In order to define an estimator or at least an upper bound of the number of $k$-colorings, we need to have a large number of graph instances for which we know the exact number of $k$-colorings. Fabio Furini et al.~\cite{Furini2017} have published an open-source and very efficient version of the backtracking DSATUR algorithm~\cite{Brelaz79} which returns the chromatic number of a given graph~\footnote{Code available at: lamsade.dauphine.fr/coloring/doku.php}. DSATUR is one of the best exact algorithms for GCP, particularly for graphs with high density. We suggest readers interrested in an overview of exact methods for GCP to read~\cite{Malaguti2009,Held2012}. 

We modified their DSATUR algorithm in order to count the total number of $k$-colorings. The pseudo code of the algorithm, called CDSATUR, is presented in algorithm~\ref{algo:cdsatur}. CDSATUR returns, for all values $k$, the exact value of $\mathcal{N}(G,k)$ taking into account the permutation of colors and especially $\mathcal{N}(G,k=\chi(G))$. 

\begin{algorithm}[H]
 \DontPrintSemicolon
\Input{$G=(V,E)$ a graph and $k$ a positive integer.}
  $\mathcal{N}\leftarrow 0$\\
  $C[v]\leftarrow None,\,\forall v\in V$: $C$ is the empty coloring.\\
  $l\leftarrow0$: number of colors used by $C$.\\
  $CDSATUR(C,l)$\\
  \Return $\mathcal{N}$\\
  \SetKwFunction{FMain}{$CDSATUR$}
  \SetKwProg{Fn}{Procedure}{:}{}
  \Fn{\FMain{$C'$, $k'$}}{
  \uIf{all the vertices of $C'$ are colored}{
    \If{$k'\leq k$}{
      $\mathcal{N}\leftarrow \mathcal{N}+1$
    }
  }\Else{
    Select an uncolored vertex $v$ of $C'$\\
    \For{every feasible color $i\in \intervalleff{1}{k'+1}$}{
      $C''\leftarrow C'$, $C''[v]\leftarrow i$\\
      $k''\leftarrow max(k',i)$: number of colors used in $C''$.\\
      \If{$k''\leq k$}{
	$CDSATUR(C'',k'')$
      }
    }
  }
  
}
\caption{\label{algo:cdsatur} CDSATUR which returns the number of all $k$-colorings of $G$: $\mathcal{N}(G,k)$.}
\end{algorithm}

Fabio Furini et al. published also 2031 random GCP instances called RCBII~\footnote{Instances available in the same address} with vertices from 60 to 140 and density between 0.1 and 0.9. 
This wide variety of graphs is our reference dataset. We complete this dataset with easy DIMACS graphs~\cite{dimacs96} for which $\chi(G)$ and $\mathcal{N}(G,\chi)$ is computable with CDSATUR. 

The 2031 graphs of RCBII benchmark have characteristics described in Table~\ref{tab:dataset1}.
We can notice that $\chi(G)$ is known for all these graphs~\cite{Furini2017}. 
First we calculated $\mathcal{N}(G,\chi)$ with CDSATUR, with a time limit equals to $2400$s.
This time is enough for most of the graphs. 
There are only 210 graph instances of RBCII (on the 2031) for which CDSATUR does not have enough time to find $\mathcal{N}(G,\chi)$.
These 210 graphs are used to test our approach (\emph{test dataset}).  

Among the graphs for which $\mathcal{N}(G,\chi)$ can be determined, we consider only those with less than 1 million optimal solutions: they form the \emph{reference dataset} (959 graph instances).
Finally, we can distinct inside the reference dataset, graph instances verifying $i(G)>\mathcal{N}(G,\chi)$ (566/959) or not (393/959). 

It remains 862 graphs on the 2031 of RBCII benchmark with more than 1 million of optimal solutions. We decided to test our approach on those graphs (called \emph{control dataset}) to check if the proposed algorithm can produce \emph{false positives} or not.

\begin{table*}[htb]
\caption{\label{tab:dataset1}Distribution of 2031 RCBII graph instances}
    \begin{center}
\small
	\scriptsize
	\begin{tabular}{cc|cc|cc|cc||c}
\toprule 
 \multicolumn{8}{c||}{$\chi$ known}  & Total\\
\multicolumn{2}{c|}{$\mathcal{N}>10^6$ } & \multicolumn{4}{c|}{$\mathcal{N}\leq 10^6$} & \multicolumn{2}{c||}{$\mathcal{N}$ ?} & \string#instances\\
\midrule
\multicolumn{2}{c|}{862 (control dataset)}&  \multicolumn{4}{c|}{959 (reference dataset)}&  \multicolumn{2}{c||}{210 (test dataset)}& 2031\\
\multicolumn{2}{c|}{}&  \multicolumn{2}{c}{$i(G)\leq\mathcal{N}$}&\multicolumn{2}{c|}{$i(G)>\mathcal{N}$}&&& \\
\multicolumn{2}{c|}{}&  \multicolumn{2}{c}{393}&\multicolumn{2}{c|}{566}& & & \\
opt. clue&not opt. clue&opt. clue&not opt. clue&opt. clue&not opt. clue&opt. clue&not opt. clue&\\
0&862& 0& 393&449&117& 39&171& \\
\bottomrule
	\end{tabular}
\end{center}
\end{table*}

\subsection{Analysis of graph instances}
Before determining an upper bound of $\mathcal{N}(G,\chi)$, 
we investigate the possible links between standard features of a graph as its size (number of vertices), its density, 
or its chromatic number and the number of optimal colorings:  $\mathcal{N}(G,\chi)$

\subsubsection{Links between $\mathcal{N}(G,\chi)$, graph size and density} 


Graphs with same size (number of vertices) and same density can have a number of optimal colorings very different from one another.
A typical example is given in Figure~\ref{fig:histo} where is represented the distribution of 49 graph instances with 80 vertices and density 0.3 (<r80\_30.*> of RBCII benchmark) in function of the number of solutions $\mathcal{N}(G,\chi)$.
Half of the graphs (25/49) have less than 100 000 optimal solutions while a third  (18/49) have more than 1 million  optimal solutions. 
There are no simple law that characterize this distribution. 
\begin{figure}
 \begin{minipage}[b]{.48\linewidth}
  \centering\epsfig{figure=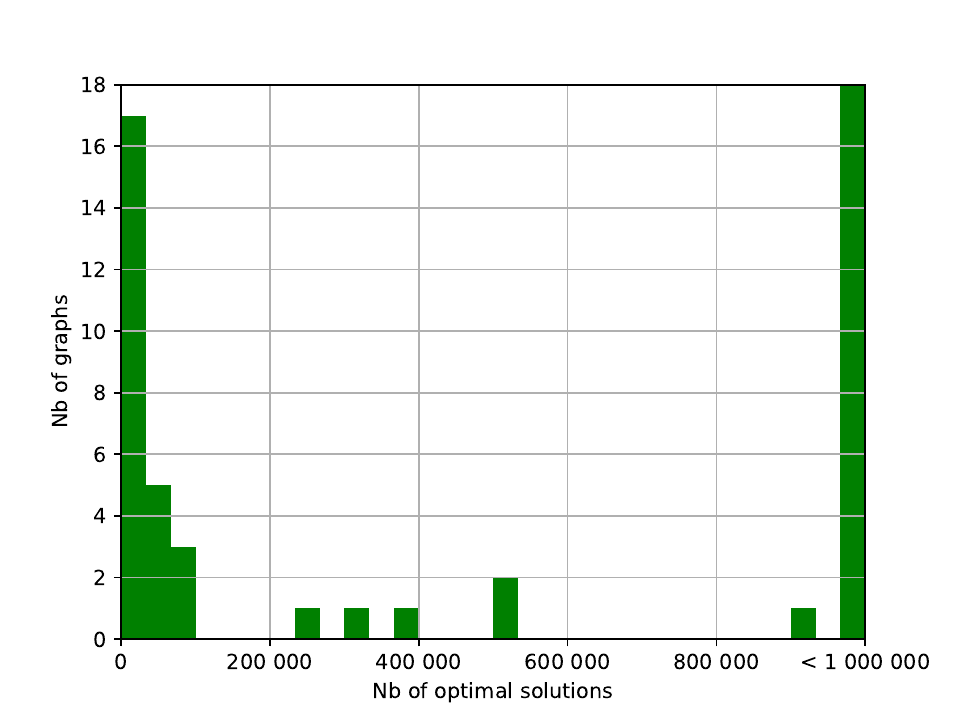,width=0.9\linewidth}
\caption{\label{fig:histo}Histogram characterizing the distribution of the 49 random graphs with 80 vertices and a density of 0.3 (<r80\_30.*> for the RBCII benchmark) given the number of optimal colorings, $\mathcal{N}(G,\chi)$.}
 \end{minipage} \hfill
 \begin{minipage}[b]{.48\linewidth}
  \centering\epsfig{figure=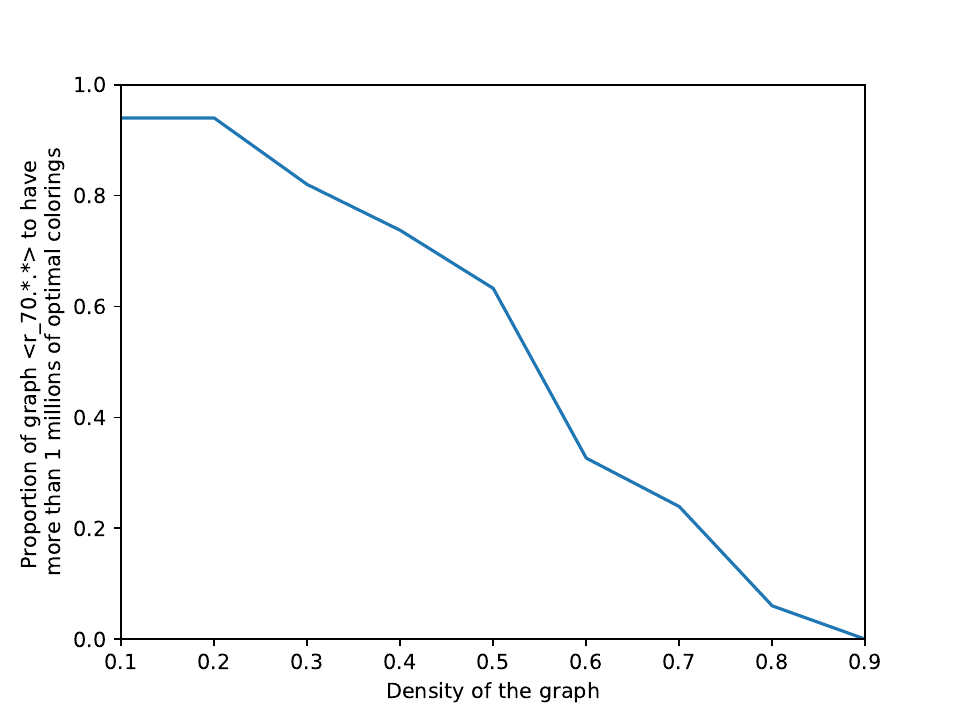,width=0.9\linewidth}
\caption{\label{fig:impactdensity}Proportion of graphs with 70 vertices (<r70\_*.*> of RBCII benchmark) having more than 1 million of optimal colorings given the density.}
 \end{minipage}
\end{figure}

However, we can notice that the lower the density, the higher the optimal solution number. Indeed, Figure~\ref{fig:impactdensity} presents the proportion of graphs with 70 vertices of RBCII benchmark having more than 1 million colorings depending on graph density. For a low density such as 0.1, nearly all graphs have more than 1 million optimal solutions, while no graph with high density (equals to 0.9).

In order to have a more fine view of the link between the number of optimal colorings and the graph density, 
we generated 1,000 random graphs with 50 vertices and density $d$ ($d=0.1$, $0.2,...,$ or $0.9$).
Each line in Figure~\ref{fig:identity} represents (for each density) the proportion of graphs having less than $n$ optimal colorings with $n$ between $10^2$ and $10^6$.
Pink line of Figure~\ref{fig:identity} shows for example that 50\% of graphs (with 50 vertices and density = 0.3) have less than $10^5$ optimal colorings.
The plots are quite similar for graphs with 60 or 70 vertices. 
The graph size seams to have a slight influence on the number of optimal solutions.
\begin{figure}
 \begin{minipage}[b]{.48\linewidth}
  \centering\epsfig{figure=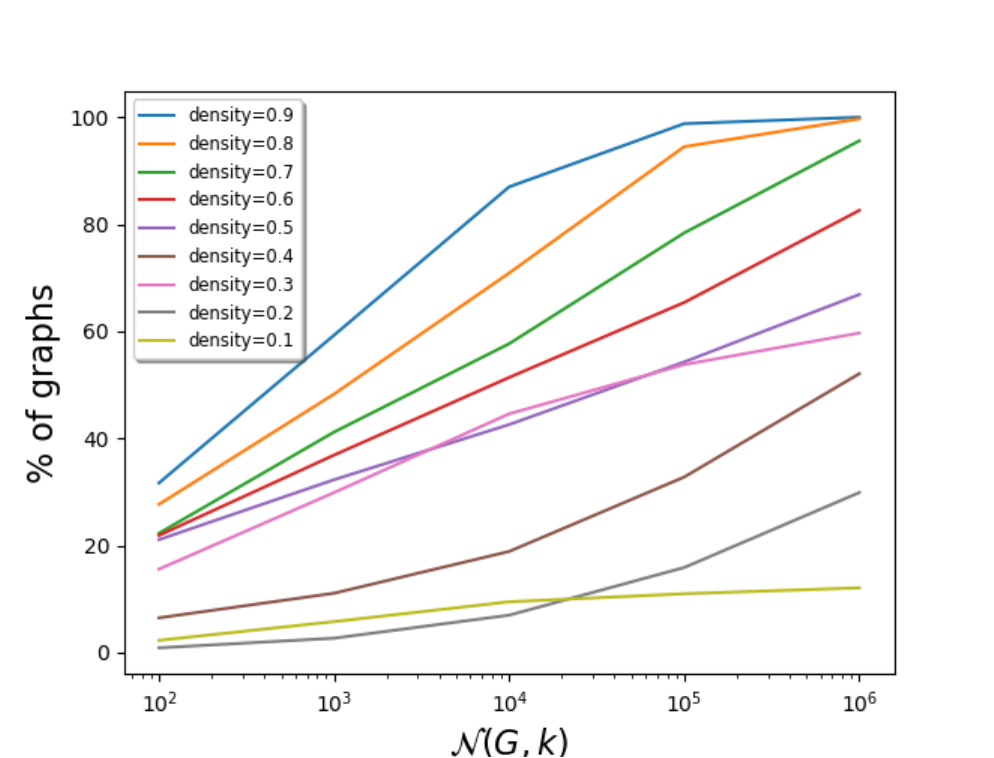,width=\linewidth}
\caption{\label{fig:identity}Proportion of random graphs (with 50 vertices and a given density) having less than $n$ optimal colorings with $n$ range between 1 hundred to 1 million.}
 \end{minipage} \hfill
 \begin{minipage}[b]{.48\linewidth}
  \centering\epsfig{figure=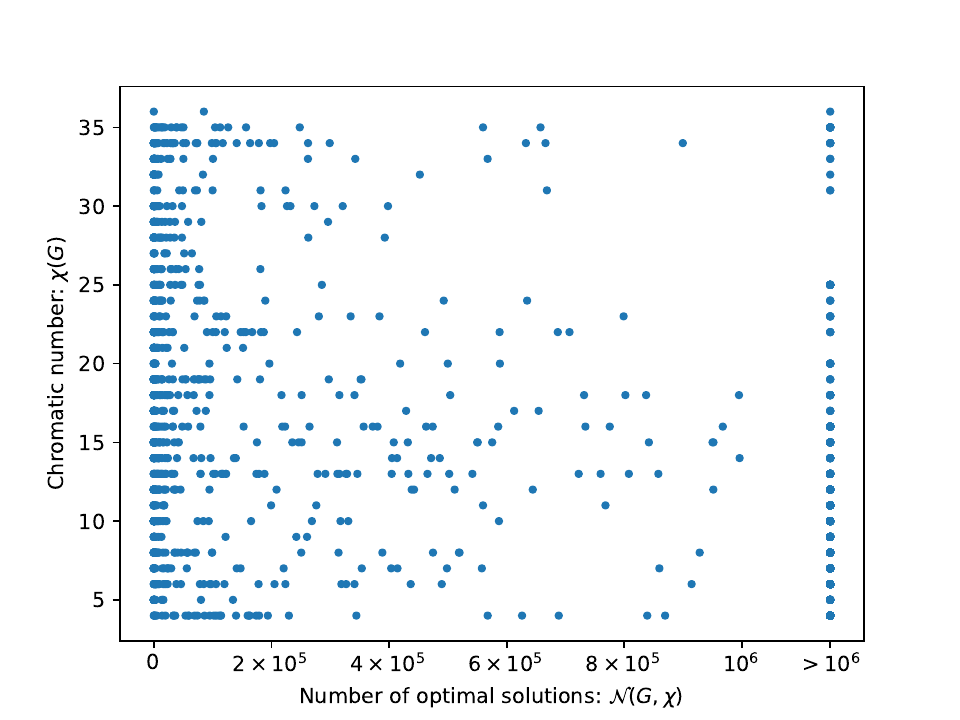,width=\linewidth}
\caption{\label{fig:link_chi_n}Each dot corresponds to one graph of RCBII. 
The abscissa is the number of optimal colorings and the ordinate is the chromatic number.}
 \end{minipage}
\end{figure}


\subsubsection{Links between $\mathcal{N}(G,\chi)$ and $\chi(G)$} 

As shown in Figure~\ref{fig:link_chi_n}, there is no obvious link between the chromatic number of a graph, 
$\chi(G)$ ($y$-axis) and the number of optimal colorings $\mathcal{N}(G,\chi)$ ($x$-axis). 
Each dot of Figure~\ref{fig:link_chi_n} corresponds to one graph of RCBII for which it is possible to calculate exactly $\mathcal{N}(G,\chi)$ with CDSATUR.

\subsection{Upper bound function}
\label{sect:upperbound}
We define in this Section an upper bound of $\mathcal{N}(G,k)$ based on the 953 graphs of the reference dataset.
Suppose we have, for a given graph $G$, a set of $n$ different $k$-colorings: 
$\Omega(G,k)=\{x_1,...,x_n\}$, i.e. $n=|\Omega(G,k)|=\mathcal{N}(G,k)$ is unknown. 
We also have a sequence $W$ of $t$ independent samples: $W=(w_1,...,w_t)$, where $w_k\in \Omega(G,k),\ \forall k=1...t$.
This sample $W$ is composed 
of $t$ independent success runs of \emph{HEAD} algorithm. 
We note $\forall j=1...n,\ \string#(x_j)$ the count of $x_j$ in $W$.
For these $t$ colorings, we count $p$ different colorings in $W$: $p=|\{x_j\in W,\ \string#(x_j)>0\}|$. 
So then, $\mathcal{N}(G,k)\geq p$ and $t\geq p\geq 1$. 
Figures~\ref{fig:different_colorings} represent for each graph of the reference dataset, 
the number of \textbf{different} colorings $p$ found by \emph{HEAD} 
on the total of $t=1,000$ success runs (in abscissa) and the exact number of colorings, $\mathcal{N}(G,k)$, 
calculated with CDSATUR (in ordinate). 
Each dot corresponds to one graph of the reference dataset. 
The objective now is to determine an upper bound of $\mathcal{N}(G,k)$, $UB$, as small as possible. 
Indeed, in order to apply the Theorem~\ref{thm:opt}, we must have $i(G)-k>UB$.

\begin{figure}
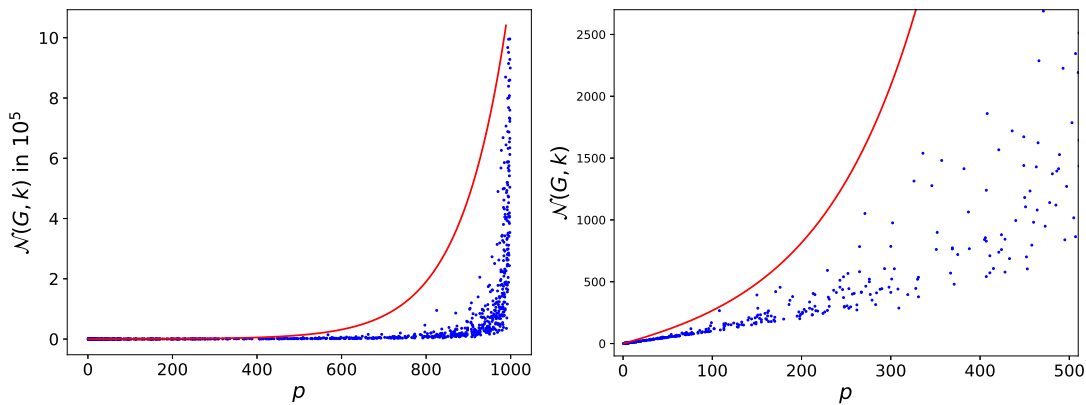

\begin{center}
\includegraphics[width=0.45\linewidth]{nbSol_nbSolDiff2.pdf}
\includegraphics[width=0.45\linewidth]{nbSol_nbSolDiff_zoom2.pdf}
\caption{\label{fig:different_colorings}Each blue dot corresponds to one of the 959 graph instances 
that have less than 1 million optimal colorings (reference dataset). 
The number of optimal colorings (calculated with CDSATUR algorithm) is in ordinate. 
The number $p$ of different optimal solutions found by our \emph{HEAD} algorithm after 1,000 success runs is in abscissa. 
The red line is an upper bound $UB(G,k,p,t)$ of the number of optimal colorings.
The right figure is a zoom of left figure for $p\leq500$.}
\end{center}
\end{figure} 

Figure~\ref{fig:different_colorings}-right which is a zoom of the left figure for $p\leq500$ shows that 
for $p\ll t$, $p$ is near linear to $\mathcal{N}(G,k)$: $p\sim\mathcal{N}(G,k)$. 
Then, $p$ is a good candidate to be an estimator of $\mathcal{N}(G,k)$.
When $p$ is near to $t$, the range of $\mathcal{N}(G,k)$ values is very large, near to $p^2=10^6$, 
and $p$ is a very bad estimation of $\mathcal{N}(G,k)$, but notice that $\mathcal{N}(G,k)<p^2$.
We add on those figures a red line that represents a possible upper bound of $\mathcal{N}(G,k)$ which is equal to:
\begin{equation}\label{eq:upperbound}
UB(G,k,p,t)=\left\{\begin{array}{ll}
p+p^{\alpha\frac{t+p}{t}}&\text{if } p<t \times 0.99\\
+\infty&\text{otherwise}
\end{array}\right.
\end{equation}
with $\alpha=1.01$.
Indeed, when $p\ll t$, $UB(G,k,p,t)\sim2p$ and when $p$ is near to $t$, $UB(G,k,p,t)\sim p^2$. 
Between these extreme values, the cloud of blue dots follows approximately an exponential curve.  
$UB(G,k,p,t)$ was also built to be above all blue dots; i.e. it is a valid upper bound for all graphs of the reference dataset. 
Of course, there is no guarantee that this upper bound is still valid for all other graphs. 
So, our approach is never able to prove optimality in a strict sense.
It gives only a clue.


\section{Experiments and analysis}\label{sec:test}

\subsection{Tests}
The upper bound $UB$ was built based on the graphs of the \emph{reference dataset}. 
Now, in order to test the \emph{optimality clue} (procedure Section~\ref{sec:procedure}), 
we use this upper bound on graphs of the \emph{test dataset} 
and the \emph{control dataset} and for some graphs coming from the DIMACS benchmark.


Results on RCBII benchmark are presented in Table~\ref{tab:dataset1} in the two last lines. 
The first column concerns the 862 graphs with more than 1 million optimal solutions, corresponding to the \emph{control dataset}.
There is no \emph{false positive}: the procedure~\ref{sec:procedure} concludes for all the graphs that there is no optimality clue.
The two following columns concern the \emph{reference dataset}.
More precisely, the second column concerns graphs having less than 1 million optimal solutions but 
that do not verify Theorem~\ref{thm:opt}: the number of IS is lower than the number of optimal solutions. 
Of course, there are no \emph{false positives} for this case, because $UB$ was built to validate those graphs (\emph{reference dataset}).
The third column concerns the 566 graphs verifying  the Theorem~\ref{thm:opt}.
The \emph{optimality clue} is proven for 449 of them because $i(G)>UB(G,k,p,t)>\mathcal{N}(G,k)$. 
The optimality clue is not shown on the 117 ($=566-449$) other graphs because $UB(G,k,p,t)\geq i(G)>\mathcal{N}(G,k)$. 
$UB(G,k,p,t)$ is an upper bound too high in this case. 
To prove the optimality clue on those graphs, we would have to increase the size of the solutions sample, $t$. 
The fourth column concerns the \emph{test dataset} i.e. graphs for which the number of optimal solutions is unknown. 
We prove the \emph{optimality clue} for nearly 20\% of these graphs (39/210). 
There are three reasons why we did not prove the optimality clue for the other 171 (=210-39) graphs: 
1) graph instances have more than 1 million solutions; 
2) graph instances do not verify the Theorem~\ref{thm:opt}; 
Nothing can be done for these two first reasons.
3) $p$ is too close from $t$, then the upper bound $UB$ is too high. 
In order to have an upper bound more accurate,
i.e. still valid but not too high, 
we have to increase the size of the sample or to choose another formula than equation~(\ref{eq:upperbound}).
Our approach therefore applies to about 20\% of the random graphs in the RCBII benchmark.
For control and reference datasets, we get more or less the same proportion: 25\% (449/1821).

\begin{table*}[htb]
\caption{\label{tab:datasetDIMACS}Results of \emph{optimality clue} tests for graphs of DIMACS benchmark with $p<t$.}
\small
\tiny
\hspace*{-1.2cm}
    \begin{tabular}{cccccccccccc||cr}
\toprule 
Instances & $|V|$ & $d$& $\chi(G)$&$k$&  $i(G)$& $\mathcal{N}(G,k)$ &$t$& $p$ & $UB(G,k,p,t)$ &Opt. clue& time (s) & $\underline{\chi}(G)$& time(s)\cite{Held2012}\\
\midrule
\bottomrule

\input{tab_dimacs_res2}

	\end{tabular}
\end{table*}

The results on selected DIMACS benchmark graphs are presented in Table~\ref{tab:datasetDIMACS}. 
We only present graph instances for which the solutions sample generated by \emph{HEAD}, are not all different (i.e. $p<t=1,000$) 
and are susceptible to be optimal. 
The first column of Table~\ref{tab:datasetDIMACS} indicates the name of the graph instance. 
Columns 2-7 indicate for each graph, its number of vertices $|V|$, its density 
$d$, its chromatic number $\chi(G)$, when it is known, the number of colors $k$, 
used for the test ($k=\chi(G)$ if $\chi(G)$ is known), its number of independent sets 
$i(G)$, and the exact number of legal $k$-colorings $\mathcal{N}(G,k)$, when it is possible to calculate it with CDSATUR. 
Then, columns 8-10 indicate the size of the solutions sample $t$, 
the number of different solutions in the sample $p$, 
the experimental upper bound of the number of $k$-colorings $UB(G,k,t,p)$. 
Columns 11 and 12 provide the result of the procedure of Section~\ref{sec:procedure} of \emph{optimality clue} and the total computation time in seconds to generate all the sample.
The two last columns indicates the lower bound of $\chi(G)$ found by the best known exact method~\cite{Held2012} 
(or by IncMaxCLQ~\cite{Li2013}, that found the maximum clique) and the computation time of this method.

The first part of Table~\ref{tab:datasetDIMACS} corresponds to 17 graphs for which $\chi(G)$ is already known by other methods. We prove the \emph{optimality clue} for 12 of them. Computation time of optimality clue is higher than those for finding the upper bound with the exact methods except for two graphs. However the computation time of optimality clue can be  considerably reduced because the 1,000 runs of HEAD can be switch on 2, 3... or 1,000 different processors or on a cluster of computers.

The second part of Table~\ref{tab:datasetDIMACS} (below the horizontal line) corresponds to 6 graphs for which $\chi(G)$ is unknown. We have generated 3 new graph instances called <DSJC+*.*\_k> with rules almost similar of those of <DSJC*.*> but for which a $k$-coloring is hidden and so that very few others $k$-colorings can exist.
For <DSJC500.5> which is one of the most challenging graph of DIMACS, 
we have the \emph{optimality clue} for $k=47$\footnote{For <DSJC500.5> the computation time is not report because it takes several weeks and no accurate time has been recorded.}.

Notice that when the optimality clue is proven for a given $k$, 
we check that the optimality clue for $k+1$ can not be proved as well.
For example, for <DSJC500.5>, we check the optimality clue for $k=48$ is not achieved: 
for $t=100, 000$ $48$-colorings found by HEAD, all are different. 
If the randomized algorithm \emph{HEAD} is biased, 
i.e. finds $k$-colorings always in the same subset of $\Omega(G,k)$ 
(and therefore undervalues the upper bound), 
this bias does not reveal for $k+1$. 

Notice that the impact of the size of the sample has a great importance on the test.
For having a not too high upper bound of  $\mathcal{N}(G,k)$, we should have $p\ll t$. 
So, the size $t$ of the solutions sample can be chosen in function of the results of each graph instance. 
For example, we can extend the sample until all colorings of the sample are found at least twice (cf. Good-Turing estimator). 




\section{Conclusions and perspectives}\label{sec:conclusion}

Based on Theorem~\ref{thm:opt}, we propose a procedure, called \emph{optimality clue}, for determining if the global optimum is reached or not, by a heuristic method.
This approach estimates an upper bound of the number of legal $k$-colorings by running a randomized heuristic several times. 
This process is contextual to the instance to solve. No general conclusion can be drawn on the heuristic itself, which is used to build solutions. 
This definition can be seen as an experimental criterion that evaluates the convergence of a randomized algorithm 
to the chromatic number. 
However, since it is not possible to be sure that the upper bound is exact, 
it is not possible to prove optimality in the strict sense. 

Our approach is nevertheless an alternative when the exact methods are not applicable (high optimality gap).
It is a new way for providing a criterion on the \emph{proximity} of the optimality.
The general idea is that the number of solutions with a same objective function value decrease when the objective function is getting closer to the optimal value.
\emph{Optimality clue} matches with the standard definition of optimality on a wide number of instances of DIMACS and RBCII benchmarks where the optimality is known. 
Furthermore, we proved the \emph{optimality clue} for <DSJC500.5> graph of DIMACS with $k=47$ colors 
which is a very challenging instance (only two algorithms are able to find $47$-colorings~\cite{Titiloye12,Moalic2018}).
Tests on small random graphs (under 140 vertices) show that \emph{optimality clue} can be proved for 20\% of them.

Finally, we defined an upper bound quite high to avoid \emph{false positives}: 
graphs for which we prove the optimality clue for a given $k$, while $\chi(G)\neq k$.\footnote{In this context, we propose on our website a challenge to find a counterexample (false positive graph)}.


\paragraph{Representative sample} The proposed approach is based on a sampling of the legal $k$-colorings space, $\Omega(G,k)$.
This sampling is built by running many times \emph{HEAD} algorithm; each success run providing one element of the sample.
Ideally, to obtain a representative sample, 
\emph{HEAD} has to uniformly draw one $k$-coloring inside the legal $k$-colorings space. 
Of course, it is not possible to guarantee this feature in all cases, 
it is why we built an upper bound function (Equation~\ref{eq:upperbound}) of $\mathcal{N}(G,k)=|\Omega(G,k)|$. 
In order to improve our approach and to get closer to the ergodic objective, 
we plan to use more powerful model counting such as presented in~\cite{Gomes2009,Ermon2012} 
and study the ergodic propriety of \emph{HEAD}. 
Our work is only an initial contribution to the study of optimality by counting.
Other methods of estimating the population should be tested, 
such as Good-Turing methods that estimate missing mass (i.e. missing $k$-colorings in the sample) 
or Peterson-type methods to obtain statistical guarantees.
\paragraph{Generalization to other optimization problems}
All tests presented in this paper are done on GCP, which has the special propriety of Lemma~\ref{lemme:iG}. 
Then, it is possible to generalize our approach to other problems, as soon as they have an analogue propriety of Lemma~\ref{lemme:iG}. 
This is the case for the maximum clique problem.

%% file: tab_dimacs_res2.tex
DSJC125.5 & 125 & 0.5 & 17 & 17 & 537,508 & ? & 1,000 & 767 & 141,503 & \textbf{True}&161&17&274\\
DSJC125.9 & 125 & 0.9 & 44 & 44 & 1,249 & ? & 1,000 & 998 & $+\infty$ & False&28&44&7\\
DSJC250.9 & 250 & 0.9 & 72 & 72 & 6,555 & ? & 1,000 & 889 & 423,733 & False&1,963&72&11,094\\
flat1000\_50\_0 & 1,000 & 0.49 & 50 & 50 & $>10^7$ & ? & 1,000 & 1 & 2 & \textbf{True}&25,694&50&3,331\\
flat1000\_60\_0 & 1,000 & 0.49 & 60 & 60 & $>10^7$ & ? & 1,000 & 1 & 2 & \textbf{True}&44,315&60&29,996\\
le450\_5a & 450 & 0.06 & 5 & 5 & $>10^7$ & 32 & 1,000 & 32 & 69 & \textbf{True}&60&5&<0.1\cite{Li2013}\\
le450\_5b & 450 & 0.06 & 5 & 5 & $>10^7$ & 1 & 1,000 & 1 & 2 & \textbf{True}&138&5&<0.1\cite{Li2013}\\
le450\_5c & 450 & 0.1 & 5 & 5 & $>10^7$ & 1 & 1,000 & 1 & 2 & \textbf{True}&28&5&<0.1\cite{Li2013}\\
le450\_5d & 450 & 0.1 & 5 & 5 & $>10^7$ & 8 & 1,000 & 8 & 16 & \textbf{True}&20&5&<0.1\cite{Li2013}\\
le450\_15c & 450 & 0.17 & 15 & 15 & $>10^7$ & ? & 1,000 & 919 & 554,866 & \textbf{True}&&15&<0.1\cite{Li2013}\\
le450\_15d & 450 & 0.17 & 15 & 15 & $>10^7$ & ? & 1,000 & 579 & 26,041 & \textbf{True}&&15&<0.1\cite{Li2013}\\
myciel3 & 11 & 0.36 & 4 & 4 & 102 & 520 & 1,000 & 435 & 7,105 & False&10&4&<0.1\\
queen5\_5 & 25 & 0.53 & 5 & 5 & 461 & 2 & 1,000 & 2 & 4 & \textbf{True}&9&5&<0.1\cite{Li2013}\\
queen6\_6 & 36 & 0.46 & 7 & 7 & 2,634 & 20 & 1,000 & 20 & 42 & \textbf{True}&10&7&<0.1\\
queen7\_7 & 49 & 0.4 & 7 & 7 & 16,869 & 4 & 1,000 & 4 & 8 & \textbf{True}&10&7&<0.1\cite{Li2013}\\
queen8\_8 & 64 & 0.36 & 9 & 9 & 118,968 & $>$154,068 & 1,000 & 993 & $+\infty$ & False&11&9&<1\\
r125.1c & 125 & 0.97 & 46 & 46 & 787 & ? & 1,000 & 977 & 934,514& False&5,962&46&<0.1\cite{Li2013}\\\hline
DSJC250.5 & 250 & 0.5 & ? & 28 & 24,791,612 & ? & 1,000 & 999 & $+\infty$ & False&1,696&26&18\\
DSJC500.5 & 500 & 0.5 & ? & 47 & $>10^7$ & ? & 341 & 281 & 32,731 & \textbf{True}&out of time&43&439\\
          &     &     &   & 48 &             & ? & 100,000 & 100,000 & $+\infty$ & False&&&\\
DSJC500.9 & 500 & 0.9 & ? & 126 & 35,165 & ? & 1,000 & 927 & 59,623 & False& 234,496 &123&100\\
DSJC+300.1\_8 & 300 & 0.1 & ? & 8 & $>10^7$ & ? & 1,000 & 3 & 6 & \textbf{True} & 22,896 &5&<0.1\cite{Li2013}\\
DSJC+300.5\_31 & 300 & 0.5 & ? & 31 & $>10^7$ & ? & 1,000 & 2 & 4 & \textbf{True} & 69,363 &29&20\\
DSJC+400.5\_39 & 400 & 0.5 & ? & 39 & $>10^7$  & ? & 1,000 & 96 & 252 & \textbf{True} &  386,037 &36&135\\

%% file: optimality_appendix.tex
\newpage
\clearpage
\section*{Appendix: Generalization to other optimization problems}

Here is a presentation of the conditions to be fulfilled for a generalization of the proposed approach to other optimization problems.
Then, we recall some basic notations. 
An optimization problem can always be modeled as follows:
\begin{equation*}
\langle f,X\rangle\left\{
\begin{array}{rl}
\displaystyle \text{Minimize}   & f(x)  \\[1mm]
\text{s.c.}                      & x\in X
\end{array}
\right. 
\end{equation*}
where $X$ is the set of admissible solutions, $x$ the decision variable and $f:X\rightarrow\mathbb{R}$ is the objective function. We note $\langle f,X\rangle$ this problem, $\mathcal{S}(\langle f,X\rangle,k)=\{x\in X | f(x)=k\}$ the set of admissible solutions with an objective function value equal to $k$ and $\mathcal{N}(\langle f,X\rangle,k)=|\mathcal{S}(\langle f,X\rangle,k)|$ the size of this set, i.e. the number of admissible solutions with objective function value equal to $k$.
For the GCP, $X$ corresponds to the set of legal colorings and $f$ provides the number of colors used.

\begin{Definition}\label{def:exeriment_opt}
Knowing one instance of an optimization problem (minimization case) $\langle f,X\rangle$ and $k$ a real number. If:
\begin{enumerate}
 \item $(x_1,...x_t)$ are $t$ admissible solutions 
 drawn independently from $\mathcal{S}(\langle f,X\rangle,k)$; it is a sample of $\mathcal{S}(\langle f,X\rangle,k)$,
 \item $\overline{\mathcal{N}}$ is an estimator of the upper bound of $\mathcal{N}(\langle f,X\rangle,k)$, the number of different solutions with value equal to $k$;
  this estimator is based on the sample  $(x_1,...x_t)$.
 \item $lb>0$ is a positive integer such as if $\exists\varepsilon>0,\,\mathcal{S}(\langle f,X\rangle,k-\varepsilon)\neq\emptyset$ then $\mathcal{N}(\langle f,X\rangle,k)\geq lb$; 
 that is to say, if there exists at least one solution of the optimization problem $\langle f,X\rangle$ with better objective function than $k$, then the number of solutions of $\mathcal{S}(\langle f,X\rangle,k)$ is higher than $lb$. This is analogous to Lemma~\ref{lemme:iG} for GCP.
 \item  $lb>\overline{\mathcal{N}}$
\end{enumerate}
then we say that solutions of $\mathcal{S}(\langle f,X\rangle,k)$ have a \textbf{optimality clue} of the optimization problem $\langle f,X\rangle$ relative to the upper bound $\overline{\mathcal{N}}$.
\end{Definition}
%


To estimate $\overline{\mathcal{N}}$, it is possible to adopt the same method as in Section~\ref{sec:test}. However, defining $lb$ can be quite complex or impossible depending of the problem. For example, for the Traveling Salesman Problem (TSP), if it exits a solution with a length equal to $n-\epsilon$ then, how is it possible to estimate the number of solutions with length equals to $n$? 

Optimaly clue can be easyly tested for example on Maximum Clique Problem (MCP). 
If we suspects that the maximum clique of a given graph is $\gamma(G)=29$.
A simple lower bound can be determinated.
Indeed, if a graph has an unique maximum clique of size 30, for example, then there exists at least 29 cliques with size 29, i.e.~: $lb = \gamma(G)=29$. 
This value may be small but the optimality is proved if $lb>\mathcal{N}$ with $\mathcal{N}$ the total number of clique of size 29. 

%% file: optimality_clue.bbl
\begin{thebibliography}{10}
\providecommand{\url}[1]{\texttt{#1}}
\providecommand{\urlprefix}{URL }
\providecommand{\doi}[1]{https://doi.org/#1}

\bibitem{Baillargeon2007}
Baillargeon, S., Rivest, L.P.: {Rcapture: Loglinear Models for
  Capture-Recapture in R}. Journal of Statistical Software, Articles
  \textbf{19}(5),  1--31 (2007)

\bibitem{Bollobas2001}
Bollob{\'a}s, B.: {Random Graphs}. {Cambridge Studies in Advanced Mathematics},
  Cambridge University Press, 2 edn. (2001). \doi{10.1017/CBO9780511814068}

\bibitem{Brelaz79}
Br{\'e}laz, D.: {New Methods to Color the Vertices of a Graph}. Communications
  of the ACM  \textbf{22}(4),  251--256 (1979)

\bibitem{Bron1973}
Bron, C., Kerbosch, J.: {Algorithm 457: Finding All Cliques of an Undirected
  Graph}. Commun. ACM  \textbf{16}(9),  575--577 (Sep 1973).
  \doi{10.1145/362342.362367}

\bibitem{Carraghan1990}
Carraghan, R., Pardalos, P.M.: {An exact algorithm for the maximum clique
  problem}. Operations Research Letters  \textbf{9}(6),  375--382 (1990).
  \doi{10.1016/0167-6377(90)90057-C}

\bibitem{Ermon2012}
Ermon, S., Gomes, C.P., Selman, B.: {Uniform Solution Sampling Using a
  Constraint Solver As an Oracle}. In: de~Freitas, N., Murphy, K.P. (eds.)
  {Proceedings of the Twenty-Eighth Conference on Uncertainty in Artificial
  Intelligence, Catalina Island, CA, USA, August 14-18, 2012}. pp. 255--264.
  {AUAI} Press (2012)

\bibitem{Favier2011}
Favier, A., de~Givry, S., J{\'e}gou, P.: {Solution counting for CSP and SAT
  with large tree-width}. Control Systems and Computers  \textbf{2},  4--13
  (mar 2011)

\bibitem{Frieze2007}
Frieze, A., Vigoda, E.: {A Survey on the use of Markov Chains to Randomly
  Sample Colourings}, chap.~4. Oxford University Press, Oxford (2007).
  \doi{10.1093/acprof:oso/9780198571278.003.0004}

\bibitem{Furini2017}
Furini, F., Gabrel, V., Ternier, I.: {An Improved DSATUR-Based Branch-and-Bound
  Algorithm for the Vertex Coloring Problem}. Networks  \textbf{69}(1),
  124--141 (2017). \doi{10.1002/net.21716}

\bibitem{Galinier99}
Galinier, P., Hao, J.K.: {Hybrid evolutionary algorithms for graph coloring}.
  Journal of Combinatorial Optimization  \textbf{3}(4),  379--397 (1999).
  \doi{10.1023/A:1009823419804}

\bibitem{Gomes2007}
Gomes, C.P., Hoffmann, J., Sabharwal, A., Selman, B.: {From sampling to model
  counting}. In: {In Proc. IJCAI{\rq}07}. pp. 2293--2299. IJCAI (2007)

\bibitem{Gomes2009}
Gomes, C.P., Sabharwal, A., Selman, B.: {Model Counting}. In: Biere, A., Heule,
  M., van Maaren, H., Walsh, T. (eds.) {Handbook of Satisfiability}, {Frontiers
  in Artificial Intelligence and Applications}, vol.~185, pp. 633--654. {IOS}
  Press (2009). \doi{10.3233/978-1-58603-929-5-633}

\bibitem{Grotschel1984}
Gr{\"o}tschel, M., Lov{\'a}sz, L., Schrijver, A.: {Polynomial Algorithms for
  Perfect Graphs}. In: Berge, C., Chv{\'a}tal, V. (eds.) {Topics on Perfect
  Graphs}, {North-Holland Mathematics Studies}, vol.~88, pp. 325--356.
  North-Holland (1984). \doi{10.1016/S0304-0208(08)72943-8}

\bibitem{Gusfield2002}
Gusfield, D.: {Partition-distance: A problem and class of perfect graphs
  arising in clustering}. Information Processing Letters  \textbf{82}(3),
  159--164 (2002)

\bibitem{Held2012}
Held, S., Cook, W., Sewell, E.: {Maximum-weight stable sets and safe lower
  bounds for graph coloring}. Mathematical Programming Computation
  \textbf{4}(4),  363--381 (2012). \doi{10.1007/s12532-012-0042-3}

\bibitem{Jerrum1995}
Jerrum, M.: {A very simple algorithm for estimating the number of k-colorings
  of a low-degree graph}. Random Structures \& Algorithms  \textbf{7}(2),
  157--165 (1995). \doi{10.1002/rsa.3240070205}

\bibitem{Jerrum2017}
Jerrum, M.: {Counting Constraint Satisfaction Problems}. In: Krokhin, A.,
  Zivny, S. (eds.) {The Constraint Satisfaction Problem: Complexity and
  Approximability}, {Dagstuhl Follow-Ups}, vol.~7, pp. 205--231. Schloss
  Dagstuhl--Leibniz-Zentrum fuer Informatik, Dagstuhl, Germany (2017)

\bibitem{dimacs96}
Johnson, D.S., Trick, M. (eds.): {Cliques, Coloring, and Satisfiability: Second
  {DIMACS} Implementation Challenge, 1993}, {DIMACS Series in Discrete
  Mathematics and Theoretical Computer Science}, vol.~26. American Mathematical
  Society, Providence, RI, USA (1996)

\bibitem{Karp72}
Karp, R.: {Reducibility among combinatorial problems}. In: Miller, R.E.,
  Thatcher, J.W. (eds.) {Complexity of Computer Computations}, pp. 85--103.
  Plenum Press, New York, USA (1972)

\bibitem{Krebs2009}
Krebs, C.J.: {Ecology}. Pearson, 6th edn. (2009)

\bibitem{Li2013}
Li, C., Fang, Z., Xu, K.: Combining maxsat reasoning and incremental upper
  bound for the maximum clique problem. In: 2013 IEEE 25th International
  Conference on Tools with Artificial Intelligence. pp. 939--946 (Nov 2013).
  \doi{10.1109/ICTAI.2013.143}

\bibitem{Malaguti2009}
Malaguti, E., Toth, P.: {A survey on vertex coloring problems}. International
  Transactions in Operational Research pp. 1--34 (2009)

\bibitem{Marmion2013}
Marmion, M.{\'E}., Jourdan, L., Dhaenens, C.: {Fitness Landscape Analysis and
  Metaheuristics Efficiency}. Journal of Mathematical Modelling and Algorithms
  in Operations Research  \textbf{12}(1),  3--26 (Mar 2013).
  \doi{10.1007/s10852-012-9177-5}

\bibitem{Merz2000}
Merz, P.: {Memetic algorithms for combinatorial optimization problems: fitness
  landscapes and effective search strategies}. Ph.D. thesis, Department of
  Electrical Engineering and Computer Science, University of Siegen, Germany
  (2000)

\bibitem{Miracle2016}
Miracle, S., Randall, D.: {Algorithms to approximately count and sample
  conforming colorings of graphs}. Discrete Applied Mathematics
  \textbf{210}(Supplement C),  133--149 (2016), lAGOS{\rq}13: Seventh
  Latin-American Algorithms, Graphs, and Optimization Symposium, Playa del
  Carmen, M{\'e}xico --- 2013

\bibitem{Moalic2018}
Moalic, L., Gondran, A.: {Variations on memetic algorithms for graph coloring
  problems}. Journal of Heuristics  \textbf{24}(1),  1--24 (Feb 2018).
  \doi{10.1007/s10732-017-9354-9}

\bibitem{Orlin1981}
Orlin, J., Bonuccelli, M., Bovet, D.: {An O(n2) Algorithm for Coloring Proper
  Circular Arc Graphs}. SIAM Journal on Algebraic Discrete Methods
  \textbf{2}(2),  88--93 (1981). \doi{10.1137/0602012}

\bibitem{Ostergard2002}
{\"O}sterg{\aa}rd, P.R.: {A fast algorithm for the maximum clique problem}.
  Discrete Applied Mathematics  \textbf{120}(1),  197--207 (2002).
  \doi{10.1016/S0166-218X(01)00290-6}, special Issue devoted to the 6th Twente
  Workshop on Graphs and Combinatorial Optimization

\bibitem{Pedersen2006}
Pedersen, A.S.P., Vestergaard, P.D.: {Bounds on the Number of Vertex
  Independent Sets in a Graph}. Taiwanese J. Math.  \textbf{10}(6),  1575--1587
  (12 2006)

\bibitem{Samotij2015}
Samotij, W.: {Counting independent sets in graphs}. Eur. J. Comb.  \textbf{48},
   5--18 (2015). \doi{10.1016/j.ejc.2015.02.005}

\bibitem{Shih1989}
Shih, W.K., Hsu, W.L.: {An O(n1.5) algorithm to color proper circular arcs}.
  Discrete Applied Mathematics  \textbf{25}(3),  321--323 (1989).
  \doi{10.1016/0166-218X(89)90011-5}

\bibitem{Titiloye12}
Titiloye, O., Crispin, A.: {Parameter Tuning Patterns for Random Graph Coloring
  with Quantum Annealing}. PLoS ONE  \textbf{7}(11),  e50060 (11 2012).
  \doi{10.1371/journal.pone.0050060}

\bibitem{Valiant1979}
Valiant, L.G.: {The Complexity of Enumeration and Reliability Problems}. {SIAM}
  J. Comput.  \textbf{8}(3),  410--421 (1979). \doi{10.1137/0208032}

\bibitem{Wei2005}
Wei, W., Selman, B.: {A New Approach to Model Counting}. In: Bacchus, F.,
  Walsh, T. (eds.) {Theory and Applications of Satisfiability Testing}. pp.
  324--339. Springer Berlin Heidelberg, Berlin, Heidelberg (2005)

\bibitem{Zuckerman2007}
Zuckerman, D.: Linear degree extractors and the inapproximability of max clique
  and chromatic number. Theory of Computing  \textbf{3}(6),  103--128 (2007).
  \doi{10.4086/toc.2007.v003a006},
  \url{http://www.theoryofcomputing.org/articles/v003a006}

\end{thebibliography}
